\documentclass[journal,12pt,onecolumn,draftclsnofoot]{IEEEtran}
\usepackage[margin=1in]{geometry}

\usepackage{amsmath,amssymb,amsthm,mathtools,bm,bbm}
\usepackage{graphicx}
\usepackage{algorithm}
\usepackage{algpseudocode}
\usepackage{booktabs}
\usepackage{cite}
\usepackage{url}
\usepackage[hidelinks]{hyperref}

\newtheorem{theorem}{Theorem}
\newtheorem{lemma}{Lemma}

\newtheorem{assumption}{Assumption}
\newtheorem{corollary}{Corollary}
\newtheorem{remark}{Remark}

\begin{document}

\title{Semantic Rate Distortion and Posterior Design: Compute Constraints, Multimodality, and Strategic Inference}

\author{Emrah Akyol, \IEEEmembership{Senior Member, IEEE} \thanks{This research is supported by the NSF via CCF/CIF (CAREER) \#2048042.} \thanks{Emrah Akyol is with the Electrical and Computer Engineering Department, Binghamton University, Binghamton, NY 13902 USA (e-mail: eakyol@binghamton.edu).} }

\maketitle

\begin{abstract}
We study strategic Gaussian semantic compression under rate and compute
constraints, where an encoder and decoder optimize distinct quadratic
objectives.
A latent Gaussian state generates a task-dependent semantic variable, and the
decoder best responds via MMSE estimation, reducing the encoder’s problem to
posterior covariance design under an information-rate constraint.
We characterize the strategic rate-distortion function in direct, remote, and
full-information regimes, derive semantic waterfilling and rate-constrained
Gaussian persuasion solutions, and establish Gaussian optimality under
misaligned objectives.
We further show that architectural compute limits act as implicit rate
constraints, yielding exponential improvements in semantic accuracy with model
depth and inference-time compute, while multimodal observation eliminates the
geometric-mean penalty inherent to remote encoding.
These results provide information-theoretic foundations for data- and
energy-efficient AI and offer a principled interpretation of modern multimodal
language models as posterior-design mechanisms under resource constraints.
\end{abstract}


\section{Introduction}
\label{section1}

The current era of artificial intelligence (AI) is defined by scaling laws:
increasing the amount of data, model size, and computation has been a reliable
route to improved performance and emergent capabilities. This trajectory,
however, confronts hard constraints. Data availability is finite, training and
inference compute are expensive, and energy consumption has become a central
bottleneck in deploying large-scale AI systems responsibly. These realities
call for principled statistical and information-theoretic tools that
characterize fundamental trade-offs among energy, data, computation, and
performance. This paper contributes such a toolset by developing a
rate--distortion (RD) and information design theory for semantic inference under
resource constraints.

A key obstacle to data- and energy-efficient learning is that modern AI
systems rarely aim to reconstruct raw observations. Instead, they are trained
and evaluated on objectives such as classification accuracy, decision quality,
or semantic correctness. In these settings, the relevant information is not
the observation itself but a latent task-dependent quantity. In addition, AI
systems often operate in multi-agent environments with misaligned goals:
users, platforms, and automated decision-making modules may value different
aspects of the same information. These facts challenge classical information
theory in two ways. First, standard RD theory assumes a common
distortion measure shared by encoder and decoder, whereas semantic and
task-oriented systems naturally induce distinct objectives. Second, classical
Bayesian persuasion and information design characterize how a sender can shape
a receiver's posterior beliefs, but typically neglect explicit
communication-rate and compute constraints.

This work unifies these perspectives by proposing a model of
\emph{strategic Gaussian semantic compression}. A latent Gaussian state $X$
governs the task semantics, while the encoder cares about a semantic variable
$\Theta = BX + V$, where $B$ represents a semantic transformation and $V$
captures semantic noise, preference uncertainty, or task variation. The
decoder aims to infer $X$ from a rate-limited message $M$, while the encoder
evaluates performance through a quadratic semantic loss in $\Theta$. The
interaction is modeled as a Stackelberg game: the encoder commits to an
encoding policy under a rate constraint, and the decoder best responds by
forming a Bayesian MMSE estimate $\hat X = \mathbb{E}[X\mid M]$ under its own
objective. This model naturally captures three regimes of practical interest:
(i) \emph{direct} encoding where the encoder observes $X$;
(ii) \emph{remote} encoding where the encoder observes only $\Theta$ (or an
imperfect semantic proxy); and
(iii) \emph{full-information} encoding where the encoder observes $(X,\Theta)$
jointly and can exploit semantic noise to shape the decoder's beliefs. The
third regime connects directly to Gaussian persuasion and information design.

Our main technical contribution is to show that these strategic semantic
compression problems admit a clear posterior-geometry characterization.
Because the decoder best responds via MMSE estimation, the encoder's semantic
distortion depends on the encoding mechanism primarily through the posterior
covariance of $X$ given the message. For Gaussian sources, the information
constraint becomes a log-det bound on posterior uncertainty, revealing a
fundamental entropy budget:
\begin{equation}
\log\det K_X \ge \log\det\Sigma_X - 2R,
\label{eq1}
\end{equation}
where $K_X$ is the posterior error covariance and $R$ is the available rate.

Beyond these RD and persuasion results, the framework offers a
principled lens on contemporary AI architectures and the focus on energy
efficiency. In transformer-based large language models (LLMs), there is no
explicit bit pipe, yet information flow is constrained by finite attention
bandwidth, limited embedding dimensions, bounded context windows, and
restricted inference-time compute. These architectural bottlenecks impose
implicit rate constraints on internal belief formation and posterior
refinement, analogous to the explicit rate constraints studied here. The
resulting entropy geometry yields analytic explanations for scaling phenomena:
increasing depth, width, or compute increases the effective information budget
and sharpens posterior beliefs; chain-of-thought reasoning acts as sequential
rate allocation that refines posteriors over multiple steps. Similarly, the
theory quantifies the advantage of multimodal encoders: multiple modalities
increase the recoverable semantic covariance and eliminate the geometric-mean
information collapse inherent to remote encoding. This provides a rigorous
explanation for why multimodal systems (such as vision language models)
achieve superior data efficiency without proportional increases in compute.

\smallskip

The paper makes the following contributions:
(i) it introduces a strategic Gaussian semantic compression framework linking
RD theory and information design under misaligned objectives;
(ii) it derives exact strategic rate distortion characterizations in direct
and remote observation regimes, including Gaussian optimality and semantic
waterfilling solutions;
(iii) it formulates and solves rate-constrained Gaussian persuasion in the
full-information regime, including closed-form diagonal solutions;
(iv) it quantifies performance gaps between observation regimes and shows how
multimodality closes these gaps; and
(v) it connects these information-theoretic limits to energy-constrained AI
architectures by interpreting architectural bottlenecks as implicit rate
constraints.

\smallskip

This paper is organized as follows.
Section~\ref{section2} reviews related work.
Section~\ref{section3} introduces the strategic semantic compression model and
defines the strategic RD function.
Section~\ref{section4} presents the main RD and rate-constrained
Gaussian persuasion results for direct, remote, and full-information encoders.
Section~\ref{section5} develops implications for multimodality and compute-limited
inference in learning systems.
Section~\ref{section6} discusses limitations and future directions, and
Section~\ref{section7} concludes.
%Long technical proofs  are relegated to the appendices.

\section{Related Work}
\label{section2}

This paper lies at the intersection of RD theory, strategic
communication (Bayesian persuasion), semantic compression, and modern machine
learning. We briefly review the most relevant prior work and position our
contributions relative to existing literature.

\subsection{Rate--Distortion Theory}

Rate--distortion (RD) theory provides the fundamental limits of lossy
compression under a prescribed fidelity criterion \cite{shannon1959coding}.
For Gaussian sources with quadratic distortion, the RD function admits a
closed-form solution via reverse waterfilling \cite{berger1971rate}. Vector
Gaussian RD and indirect (remote) source coding problems have also been
studied extensively \cite{wolf1970source,witsenhausen1975indirect}.
These formulations assume a common distortion measure shared by encoder and
decoder and do not address misaligned objectives. 

Our work departs from classical RD by allowing distinct quadratic objectives
for encoder and decoder, with the encoder's distortion defined on a semantic
variable $\Theta = BX + V$. We show that this leads to posterior-covariance
design problems with log-det entropy constraints, yielding semantic
waterfilling laws that generalize classical RD.

\subsection{Strategic Communication and Bayesian Persuasion}

Strategic communication and signaling games have a long history in economics
and control, beginning with cheap-talk models \cite{crawford1982strategic,farrell1996cheap}. 
%In linear-quadratic-Gaussian (LQG) settings, Stackelberg equilibria between
%encoder and decoder with misaligned objectives have been characterized,
%including the optimality of linear signaling policies under certain
%conditions \cite{akyol2016information}. 
A special case where a sender commits to an information structure
to influence a receiver's posterior beliefs and actions, namely the Bayesian Persuasion setting was studied in
\cite{kamenica2011bayesian}. Subsequent work has extended these ideas to
general state spaces and to Gaussian environments with quadratic payoffs,
leading to linear-Gaussian persuasion (LQG persuasion)
\cite{gentzkow2016persuasion,tamura2018bayesian,ui2020lqg,akyol2016information}.
Signaling games have received considerable interest from the engineering as well as economics researchers, see e.g. \cite{candogan2020optimal,bergemann2019information,dughmi2019algorithmic,alonso2016persuading,gould2023information,vora2020information, kazikli2021signaling,le2019persuasion, aybas2019persuasion, sayin2019hierarchical, sayin2021bayesian,farhadi2022dynamic, bacsar2024inducement, anand2025channel}. These models (except a few that consider communication rate constraints \cite{akyol2016information,le2019persuasion} and quantization \cite{dughmi2016persuasion,aybas2019persuasion,anand2025channel, akyol2023strategic}) typically assume unrestricted signaling and do not impose
communication or rate constraints. In \cite{akyol2016information}, Gaussian persuasion with rate constraints is studied, however, only in scalar domain. The full-information setting considered in this paper can be viewed as an extension of  \cite{akyol2016information} to semantic and multidimensional settings. 

\subsection{Semantic Communication}

Recent work on semantic communication seeks to move beyond symbol-level
fidelity toward task-aware or meaning-aware transmission
\cite{liu2021rate, xie2021semantic,zhang2022semantic,10540315,9955525}. Most of these approaches demonstrate
empirical gains and motivate goal-oriented communications, however, they do not study
information-theoretic characterizations of fundamental limits or strategic misalignment between the objectives. Our results provide such a characterization in a stylized multidimensional Gaussian setting,
identifying posterior entropy as the central resource governing semantic
performance under rate and observation constraints. It also connects compute to scaling laws and multimodal observations, bridging information theory with modern ML scaling- a novelty missing in prior work.

Perhaps closest to our work, in \cite{xiao2022rate}, semantic strategic rate distortion is considered. This paper primarily extends the concepts in \cite{akyol2016information} to a few different semantic settings at the high level, however does not address specific scenarios considered in this paper. 

\subsection{Multimodal Learning and Efficient AI}

Multimodal learning combines heterogeneous observations to improve robustness
and data efficiency \cite{baltruvsaitis2018multimodal}. Empirical results show
that multimodal models outperform unimodal ones, particularly in grounding
and reasoning tasks. Separately, empirical scaling laws relating model size,
compute, and performance have been documented for large neural networks
\cite{kaplan2020scaling,hoffmann2022training}.

By interpreting architectural bottlenecks and compute limits as implicit
information-rate constraints, our framework provides an information-theoretic
explanation for these phenomena. In particular, we show that multimodal
observation eliminates a geometric-mean semantic penalty inherent to remote
encoding and that increased compute corresponds to exponential reductions in
posterior uncertainty.

\subsection{Summary}

In contrast to prior work, this paper unifies rate-distortion theory,
strategic communication, Bayesian persuasion, semantic compression, and
neural scaling laws within a single posterior-design framework. The central
object is the posterior covariance of latent semantic variables, and the
central constraint is the entropy budget imposed by rate or compute.

Unlike classical Gaussian RD which assumes aligned
objectives and focuses on reconstruction fidelity under some common distortion measure, this work studies semantic
and strategic misalignment and characterizes optimal posterior geometry under
rate constraints. Unlike LQG signaling and Bayesian persuasion models, which
typically allow unrestricted signaling, we impose explicit information-rate
constraints and derive the resulting posterior-design structure. Finally,
unlike recent semantic communication and multimodal learning approaches, which
are primarily empirical, our results provide a first-principles,
information-theoretic explanation for multimodal advantages and compute-driven
scaling laws. To the best of our knowledge, this is the first work to unify
RD theory, Gaussian persuasion, semantic objectives, and
compute-limited learning architectures within a single analytical framework.

\section{Problem Formulation}
\label{section3}

We consider a strategic communication setting in which an encoder and a
decoder possess distinct semantic objectives. A latent Gaussian state $X$
governs the semantic content of interest, while the encoder evaluates
reconstruction through a derived semantic variable $\Theta = BX + V$. The
decoder, however, seeks to estimate $X$ itself. The encoder communicates with
the decoder through a rate-limited noiseless channel and anticipates the
decoder's Bayesian best response.

{\it Notation}: All random vectors are defined on a common probability space.
All logarithms are natural logarithms, in base $e$, the rates are measured in nats.
For a random vector $Y$, we use $\mathrm{Cov}(Y)$ to denote its covariance matrix and
$$\Sigma_Y := \mathrm{Cov}(Y)$$ when convenient.
For an auxiliary random variable $U$, $\Sigma_{Y\mid U}$ denotes the (average) MMSE
error covariance:
\begin{equation}
\Sigma_{Y\mid U}
=
\mathbb{E}\!\left[(Y-\mathbb{E}[Y\mid U])(Y-\mathbb{E}[Y\mid U])^\top\right].
\label{eq2}
\end{equation}
 $A\succeq 0$ (resp.\ $A\succ 0$) means $A$ is symmetric positive semidefinite
(resp.\ definite). For symmetric $A,B$, $A\preceq B$ (Loewner order) means $B-A\succeq 0$. We let $\mathrm{tr}(\cdot)$ and $\det(\cdot)$ denote trace and determinant; 
and $I(\cdot\,;\cdot)$ and $h(\cdot)$  denote mutual information  and differential entropy respectively. 
$\mathrm{diag}(\cdot)$ and $\mathrm{blkdiag}(\cdot)$ denote diagonal and block-diagonal matrices.
$I$ denotes the identity matrix of appropriate dimension.
The expression  $Y\perp Z$ means the random variables $Y$ and $Z$  are statistically independent.

Let $X\in\mathbb{R}^k$ satisfy
\begin{equation}
X \sim \mathcal{N}(0,\Sigma_X),
\label{eq3}
\end{equation}
where $\Sigma_X \succ 0$.

The encoder evaluates reconstruction through the linear--Gaussian semantic transform
\begin{equation}
\Theta = B X + V,
\label{eq4}
\end{equation}
where $B\in\mathbb{R}^{m\times k}$ is known and
\begin{equation}
V \sim \mathcal N(0,\Sigma_V).
\label{eq5}
\end{equation}
We assume
\begin{equation}
V \perp X.
\label{eq6}
\end{equation}

The mapping $X \mapsto \Theta$ reflects the encoder’s semantic or alignment
objectives (e.g., task relevance, supervision signals).

We analyze three observation regimes. The encoder observes
\begin{equation}
Z \in \{X,\ \Theta,\ (X,\Theta)\}.
\label{eq7}
\end{equation}
\begin{itemize}
\item \textbf{Direct encoder:} $Z=X$. The encoder observes the state perfectly.
\item \textbf{Remote encoder:} $Z=\Theta$. The encoder observes only a semantic proxy of $X$.
\item \textbf{Full-information encoder:} $Z=(X,\Theta)$. The encoder observes both state and semantics.
\end{itemize}

For blocklength $n$, an $(n,R)$ source code consists of an encoder
\begin{align}
f_n &: \mathcal{Z}^n \to \{1,\dots,\exp{nR}\},
\label{eq8}\\
M &= f_n(Z^n),
\label{eq9}
\end{align}
and a decoder
\begin{align}
g_n &: \{1,\dots,\exp\{nr\}\} \to \mathbb{R}^{kn},
\label{eq10}\\
\hat X^n &= g_n(M).
\label{eq11}
\end{align}
The rate constraint is
\begin{equation}
R = \frac{1}{n}\log |{\cal M}_n|.
\label{eq12}
\end{equation}

The decoder minimizes quadratic distortion
\begin{equation}
D_d(f_n,g_n)
=
\frac{1}{n}\sum_{t=1}^n
\mathbb{E}\!\left[(X_t-\hat X_t)^\top W_d (X_t-\hat X_t)\right],
\label{eq13}
\end{equation}
with $W_d\succ0$. Because $X$ is Gaussian and the objective is strictly convex,
the decoder's unique best response is the MMSE estimator
\begin{equation}
\hat X^n = \mathbb{E}[X^n\mid M].
\label{eq14}
\end{equation}

The encoder’s distortion is
\begin{equation}
D_e(f_n,g_n^\star)
=
\frac{1}{n}\sum_{t=1}^n
\mathbb{E}\!\left[(\Theta_t - \hat X_t)^\top W_e (\Theta_t - \hat X_t)\right],
\label{eq15}
\end{equation}
with $W_e\succ0$. Here, we deliberately measure semantic mismatch against the decoder’s state estimate $\hat X$,
since in downstream decision pipelines in learning systems,  the decoder typically exposes $\hat X$ (or actions based on it),
not a separate semantic reconstruction of $\Theta$.

A semantic distortion level $D_e$ is achievable at rate $R$ if there exists a sequence
of encoders $\{f_n\}$ such that \eqref{eq14} holds and
\begin{equation}
\limsup_{n\to\infty} D_e(f_n,g_n^\star)\le D_e.
\label{eq20}
\end{equation}

The strategic rate--distortion function is
\begin{equation}
R(D_e) = \inf\left\{R:\ D_e \text{ is achievable at rate } R\right\}.
\label{eq21}
\end{equation}

\section{Main Results}
\label{section4}

We now characterize the strategic Gaussian rate--distortion function under the
three encoder observation models introduced in Section~\ref{section3}. In the
direct and remote regimes, the problem reduces to convex optimization over
posterior covariances with a log-det rate constraint, yielding explicit
semantic waterfilling solutions. In the full-information regime, the encoder
can also manipulate posterior cross-covariances with semantic noise, leading to
a rate-constrained Gaussian persuasion problem with closed-form solutions in
the diagonalizable case.

\subsection{Direct Encoder: Observation of $X$}

We start with the setting where the encoder observes only $X$, hence any admissible message induces a posterior error
covariance $K_X$ satisfying $0\preceq K_X \preceq \Sigma_X$.

%The encoder's semantic
%distortion is given by \eqref{eq15}, when single letterized:
%$$
%D_e=\mathbb E\{(\Theta-\hat X)^{\top}W_e (\Theta-\hat X)\}
%$$

Using the orthogonal decomposition $X= \hat X + e$ with $e\perp\hat X$, the semantic
distortion can be represented as a form that we will use repeatedly for this setting. We start with lemmas that present such auxilliary results.

\begin{lemma}
\label{lemma1}
The encoder's  distortion satisfies
\begin{equation}
D_e
= \mathrm{tr}(W_e C_0) + \mathrm{tr}(\tilde W K_X).
\label{eq90}
\end{equation}
where $K_X$ denotes the posterior error covariance
\begin{equation}
K_X = \Sigma_{X\mid M},
\label{eq17}
\end{equation}
 the semantic offset covariance is given as:
\begin{equation}
C_0 = (B-I)\Sigma_X(B-I)^\top + \Sigma_V,
\label{eq18}
\end{equation}
and the effective weight matrix is
\begin{equation}
\tilde W = W_e(B-I) + (B-I)^\top W_e + W_e.
\label{eq19}
\end{equation}
\end{lemma}

\begin{proof}
Recall that $X\sim\mathcal N(0,\Sigma_X)$ and $\Theta = BX + V$ with
$V\sim\mathcal N(0,\Sigma_V)$ independent of $X$. The decoder uses
$\hat X = \mathbb E[X\mid M]$ and we define the posterior error $e := X - \hat X$.
By the MMSE orthogonality principle, $e \perp \hat X$ and
\begin{align*}
\mathrm{Cov}(e) &= K_X,\\
\mathrm{Cov}(\hat X) &= \Sigma_X - K_X.
\end{align*}

By definition,
\begin{equation*}
D_e
= \mathbb E\!\left[(\Theta - \hat X)^\top W_e(\Theta - \hat X)\right]
= \mathrm{tr}\!\left(W_e\,\mathrm{Cov}(\Theta-\hat X)\right).
\end{equation*}
Using $\Theta = BX + V = B(\hat X + e) + V$, we obtain
\begin{equation*}
\Theta - \hat X
= (B-I)\hat X + B e + V.
\end{equation*}
Because $e \perp \hat X$ and $V \perp (X,M)$, all cross-covariances among $(\hat X,e,V)$ vanish,
so
\begin{align*}
\mathrm{Cov}(\Theta-\hat X)
&= (B-I)\,\mathrm{Cov}(\hat X)\,(B-I)^\top
+ B\,\mathrm{Cov}(e)\,B^\top
+ \Sigma_V\\
&= (B-I)(\Sigma_X-K_X)(B-I)^\top
+ B K_X B^\top
+ \Sigma_V.
\end{align*}
Expanding and simplifying yields
\begin{equation*}
\mathrm{Cov}(\Theta-\hat X)
= C_0 + (B K_X + K_X B^\top - K_X).
\end{equation*}
Multiplying by $W_e$ and taking trace gives
\begin{equation*}
D_e
=
\mathrm{tr}(W_e C_0) + \mathrm{tr}\!\left((W_eB+B^\top W_e-W_e)K_X\right).
\end{equation*}
Using $\tilde W=W_e(B-I)+(B-I)^\top W_e+W_e = W_eB+B^\top W_e-W_e$ yields \eqref{eq90}.
\end{proof}

\begin{lemma}
\label{lemma2}
Let $X\sim \mathcal{N}(0,\Sigma_X)$ and let $U$ be any random variable. Define
\begin{equation}
K_X := \Sigma_{X\mid U}.
\label{eq91}
\end{equation}
Then
\begin{equation}
I(X;U)
\ge
\frac{1}{2}\log\frac{\det\Sigma_X}{\det K_X},
\label{eq92}
\end{equation}
with equality if and only if $X\mid U$ is Gaussian with covariance $K_X$ almost surely.
\end{lemma}

\begin{proof}
We have $I(X;U) = h(X) - h(X\mid U)$. Since $X\sim\mathcal N(0,\Sigma_X)$,
\begin{equation*}
h(X) = \frac{1}{2}\log\!\left((2\pi e)^k \det\Sigma_X\right).
\end{equation*}
For each $u$, let $\Sigma_{X\mid u}$ be the covariance of $X\mid U=u$. Then
\begin{equation*}
h(X\mid U=u)\le\frac{1}{2}\log\!\left((2\pi e)^k\det\Sigma_{X\mid u}\right),
\end{equation*}
with equality if and only if $X\mid U=u$ is Gaussian. Taking expectation and applying concavity of
$\log\det$ yields
\begin{equation*}
h(X\mid U)
\le
\frac{1}{2}\log\!\left((2\pi e)^k \det K_X\right),
\end{equation*}
which implies \eqref{eq92}. Equality requires Gaussianity and constant conditional covariance.
\end{proof}

\begin{lemma}
\label{lemma3}
Let $X\sim\mathcal N(0,\Sigma_X)$ and let $0\prec K_X\preceq\Sigma_X$. Define
\begin{equation}
\Sigma_Z^{-1} := K_X^{-1} - \Sigma_X^{-1}.
\label{eq93}
\end{equation}
Let $Z\sim\mathcal N(0,\Sigma_Z)$ be independent of $X$ and define $U=X+Z$.
Then $\Sigma_{X\mid U}=K_X$ and
\begin{equation}
I(X;U)=\frac12\log\frac{\det\Sigma_X}{\det K_X}.
\label{eq94}
\end{equation}
\end{lemma}

\begin{proof}
For jointly Gaussian $(X,U)$ with $U=X+Z$, the posterior covariance is
\begin{align*}
\Sigma_{X\mid U}= (\Sigma_X^{-1}+\Sigma_Z^{-1})^{-1}= \left(\Sigma_X^{-1} + K_X^{-1} - \Sigma_X^{-1}\right)^{-1}= K_X.
\end{align*}
Mutual information satisfies $I(X;U)=h(X)-h(X\mid U)$, which yields \eqref{eq94}.
\end{proof}

\begin{theorem}
\label{thm1}
The strategic RD function when the encoder observes $X$ is
\begin{equation}
R(D_e)
=
\min_{0\preceq K_X \preceq\Sigma_X}
\left\{
\frac12\log\frac{\det\Sigma_X}{\det K_X}
:\;
\mathrm{tr}(\tilde W K_X) \le D_e'
\right\},
\label{eq22}
\end{equation}
where $D_e' := D_e - \mathrm{tr}(W_e C_0)$. For every feasible $K_X$, Gaussian test
channels $U=X+Z$ achieve rate $\frac12\log\frac{\det\Sigma_X}{\det K_X}$ and induce
$\Sigma_{X\mid U}=K_X$.
\end{theorem}

\begin{proof}

\emph{Achievability:} Fix any $K_X$ with $0\preceq K_X\preceq\Sigma_X$ and
$\mathrm{tr}(\tilde W K_X)\le D_e'$. By Lemma~\ref{lemma3} there exists a Gaussian
encoder $U=X+Z$ that induces $\Sigma_{X\mid U}=K_X$ and achieves
\begin{equation*}
R = I(X;U)=\frac12\log\frac{\det\Sigma_X}{\det K_X}.
\end{equation*}
Block coding according to the test channel achieves distortion arbitrarily close to
$\mathrm{tr}(W_eC_0) + \mathrm{tr}(\tilde W K_X)$ and rate arbitrarily close to $R$.
Taking the infimum over all feasible $K_X$ yields \eqref{eq22}.

\emph{Converse:} Any code induces a posterior covariance $K_X$ satisfying
$0\preceq K_X\preceq\Sigma_X$ and distortion \eqref{eq90}. Moreover, by Lemma~\ref{lemma2},
\begin{equation*}
R \ge I(X;M)\ge \frac12\log\frac{\det\Sigma_X}{\det K_X}.
\end{equation*}
Therefore, any achievable $(R,D_e)$ must satisfy the constraints in \eqref{eq22} for some such $K_X$.

\end{proof}
\begin{remark}
In unconstrained  information design 
the sender chooses an information structure that induces a distribution over posterior
means $\hat X=\mathbb E[X\mid M]$ subject to Bayes plausibility \cite{kamenica2011bayesian}. For the Gaussian setting at hand, this implies second order moment conditions
(see e.g., \cite{tamura2018bayesian,ui2020lqg}),
\begin{equation}
0 \preceq \Sigma_{\hat X\hat X} \preceq \Sigma_X,
\label{eq23}
\end{equation}
together with
\begin{equation}
K_X=\Sigma_X-\Sigma_{\hat X\hat X}.
\label{eq24}
\end{equation}
Since the sender’s quadratic objective is affine in these covariance blocks, the
unconstrained problem is an SDP, as shown in \cite{tamura2018bayesian}. 

Imposing an explicit communication-rate constraint $I(X;M)\le R$ shrinks the Bayes-plausible
set by enforcing an entropy budget on the posterior error covariance:
\begin{equation}
\log\det K_X \ge \log\det\Sigma_X - 2R,
\label{eq25}
\end{equation}
which is the same constraint highlighted in \eqref{eq1}.
\end{remark}

%When $\Sigma_X$ and $\tilde W$ are jointly diagonalizable and $\tilde W\succ 0$,
%writing $K_X = \Sigma_X^{1/2}UDU^\top\Sigma_X^{1/2}$ with $D=\mathrm{diag}(d_1,\dots,d_k)$
%yields the semantic waterfilling law
%\begin{equation}
%d_i^\star = \min\left\{1,\ \frac{\nu}{\lambda_i}\right\},
%\label{eq26}
%\end{equation}
%where $\{\lambda_i\}$ are the eigenvalues of $\Sigma_X^{1/2}\tilde W\,\Sigma_X^{1/2}$ and
%$\nu$ is chosen to satisfy $\sum_{i=1}^k \lambda_i d_i^\star = D_e'$.
\begin{remark}
\label{remark3}
The program \eqref{eq22} can be written in the normalized variable
$D:=\Sigma_X^{-1/2}K_X\Sigma_X^{-1/2}$ as
\begin{equation}
\min_{0\prec D\preceq I}\;
\frac{1}{2}\log\frac{1}{\det D}
\qquad
\text{subject to}
\qquad
\mathrm{tr}(A_X D)\le D_e',
\label{eq95}
\end{equation}
where the effective normalized weight is
\begin{equation}
A_X:=\Sigma_X^{1/2}\tilde W\,\Sigma_X^{1/2}.
\label{eq96}
\end{equation}
Thus the eigen-geometry relevant to waterfilling is governed by $A_X$, not by
$\Sigma_X$ and $\tilde W$ separately. If $A_X$ has nonpositive eigenvalues,
those directions are never compressed at the rate-minimizing solution, because
setting the corresponding normalized posterior eigenvalues to one both reduces
rate and does not tighten the semantic constraint.
\end{remark}

\begin{corollary}
\label{cor4}
Let $A_X=U\Lambda U^\top$, where
\begin{equation}
\Lambda=\mathrm{diag}(\lambda_1,\dots,\lambda_k),
\label{eq97}
\end{equation}
and define the index sets
\begin{equation}
\mathcal I_+ := \{i:\lambda_i>0\},
\label{eq98}
\end{equation}
and
\begin{equation}
\mathcal I_0 := \{i:\lambda_i\le 0\}.
\label{eq99}
\end{equation}
Then there exists an optimizer of \eqref{eq22} of the form
\begin{equation}
K_X^\star
=
\Sigma_X^{1/2}U\,\mathrm{diag}(d_1^\star,\dots,d_k^\star)\,U^\top\Sigma_X^{1/2},
\label{eq100}
\end{equation}
with $0<d_i^\star\le 1$ and
\begin{equation}
d_i^\star = 1,
\label{eq101}
\end{equation}
for $i\in\mathcal I_0$, while for $i\in\mathcal I_+$
\begin{equation}
d_i^\star = \min\left\{1,\ \frac{\nu^\star}{\lambda_i}\right\}.
\label{eq102}
\end{equation}

If $D_e' \ge \mathrm{tr}(A_X)$, then $K_X^\star=\Sigma_X$ is feasible and
\begin{equation}
R(D_e)=0.
\label{eq103}
\end{equation}
If $D_e' < \mathrm{tr}(A_X)$, then the semantic constraint is active and $\nu^\star>0$
is uniquely determined by
\begin{equation}
\sum_{i\in\mathcal I_+}\min\{\lambda_i,\nu^\star\}
=
D_e' - \sum_{i\in\mathcal I_0}\lambda_i.
\label{eq104}
\end{equation}

Moreover, the solution is in the interior (no saturation on any $i\in\mathcal I_+$)
if and only if
\begin{equation}
\nu^\star < \min_{i\in\mathcal I_+}\lambda_i.
\label{eq105}
\end{equation}
In that interior regime,
\begin{equation}
\nu^\star
=
\frac{D_e' - \sum_{i\in\mathcal I_0}\lambda_i}{|\mathcal I_+|},
\label{eq106}
\end{equation}
and the resulting rate is
\begin{equation}
R(D_e)
=
\frac12\sum_{i\in\mathcal I_+}\log\left(\frac{\lambda_i}{\nu^\star}\right).
\label{eq107}
\end{equation}
\end{corollary}

\begin{remark}
\label{remark4}
Even though $W_e\succ 0$, the induced matrix $\tilde W$ need not be positive semidefinite.
If $A_X$ has negative eigenvalues, then decreasing posterior uncertainty in those directions
can increase $\mathrm{tr}(\tilde W K_X)$, so the strategic RD curve need not be monotone
in the available rate. A sufficient condition that rules out this pathology is
\begin{equation}
\tilde W\succeq 0,
\label{eq108}
\end{equation}
equivalently $A_X\succeq 0$.
\end{remark}

\subsection{Remote Encoder: Observation of $\Theta$ Only}

Suppose the encoder observes $\Theta=BX+V$, where $X\sim\mathcal N(0,\Sigma_X)$,
$V\sim\mathcal N(0,\Sigma_V)$, and $V\perp X$, and the decoder best responds with
$\hat X=\mathbb E[X\mid M]$.

Let
\begin{equation}
\Sigma_\Theta = \mathrm{Cov}(\Theta)=B\Sigma_XB^\top+\Sigma_V.
\label{eq27}
\end{equation}
Define the linear Bayes regressions
\begin{equation}
L_X=\Sigma_XB^\top\Sigma_\Theta^{-1}.
\label{eq28}
\end{equation}
\begin{equation}
L_V=\Sigma_V\Sigma_\Theta^{-1}.
\label{eq29}
\end{equation}
Let $K_\Theta$ denote the posterior error covariance of $\Theta$:
\begin{equation}
K_\Theta=\Sigma_{\Theta\mid M}.
\label{eq30}
\end{equation}

Define the infinite-rate remote baseline distortion
\begin{equation}
D_{\infty}^{\mathrm{rem}}
=
\mathbb{E}\!\left[
\big(\Theta-\mathbb{E}[X\mid\Theta]\big)^\top
W_e
\big(\Theta-\mathbb{E}[X\mid\Theta]\big)
\right],
\label{eq31}
\end{equation}
and define
\begin{equation}
Q = L_X^\top\tilde W L_X + L_X^\top W_e L_V + L_V^\top W_e L_X.
\label{eq32}
\end{equation}

\begin{theorem}
\label{thm2}
The strategic remote rate distortion function is
\begin{equation}
R(D_e)
=
\min_{0\preceq K_\Theta \preceq \Sigma_\Theta}
\left\{
\frac12\log\frac{\det\Sigma_\Theta}{\det K_\Theta}
:\;
\mathrm{tr}(QK_\Theta)\le D_e-D_{\infty}^{\mathrm{rem}}
\right\}.
\label{eq33}
\end{equation}
Gaussian test channels $U=\Theta+Z$ (with $Z\perp\Theta$ Gaussian) achieve the minimum
for every feasible $K_\Theta$.
\end{theorem}

\begin{proof}
\label{appendix2}

Since $M=f_n(\Theta^n)$ and $|{\cal M}_n|\le \exp\{nr\}$,
\begin{equation*}
nR\ge H(M)\ge I(\Theta^n;M).
\end{equation*}
By standard single-letterization (time-sharing), there exists a single-letter pair $(\Theta,U)$
with $U=(M,Q)$ such that $\frac1n I(\Theta^n;M)=I(\Theta;U)$ and $K_\Theta=\Sigma_{\Theta\mid U}$.
Applying Lemma~\ref{lemma2} to $(\Theta,U)$ gives
\begin{equation*}
R \ge I(\Theta;U)\ge \frac12\log\frac{\det\Sigma_\Theta}{\det K_\Theta}.
\end{equation*}

Let $\hat\Theta:=\mathbb E[\Theta\mid M]$ and $e_\Theta:=\Theta-\hat\Theta$, so
$\mathrm{Cov}(e_\Theta)=K_\Theta$. Because $(X,V,\Theta)$ are jointly Gaussian, conditional
expectations are linear:
\begin{align*}
\mathbb E[X\mid\Theta]&=L_X\Theta,\\
\mathbb E[V\mid\Theta]&=L_V\Theta.
\end{align*}
Since $M$ is a function of $\Theta$, the tower property yields
\begin{equation*}
\hat X=\mathbb E[X\mid M]=L_X\hat\Theta.
\end{equation*}
Define innovations $N_X:=X-L_X\Theta$ and $N_V:=V-L_V\Theta$, which are independent of $\Theta$.
Then the MMSE error is $e:=X-\hat X=N_X+L_X e_\Theta$, and one can show
\begin{equation*}
D_e = D_{\infty}^{\mathrm{rem}} + \mathrm{tr}(QK_\Theta),
\end{equation*}
with $D_\infty^{\mathrm{rem}}$ and $Q$ defined in \eqref{eq31} and \eqref{eq32}.

\emph{Converse:} The preceding steps show that any achievable $(R,D_e)$ induces
$K_\Theta$ with $0\preceq K_\Theta\preceq\Sigma_\Theta$ satisfying the constraints in \eqref{eq33}.

\emph{Achievability:} Fix any $K_\Theta$ with $0\prec K_\Theta\preceq\Sigma_\Theta$. Define
$\Sigma_Z^{-1}:=K_\Theta^{-1}-\Sigma_\Theta^{-1}$ and choose $Z\sim\mathcal N(0,\Sigma_Z)$
independent of $\Theta$. With $U:=\Theta+Z$, Gaussian identities yield $\Sigma_{\Theta\mid U}=K_\Theta$
and $I(\Theta;U)=\frac12\log\frac{\det\Sigma_\Theta}{\det K_\Theta}$. Encoding according to this test
channel achieves distortion $D_{\infty}^{\mathrm{rem}}+\mathrm{tr}(QK_\Theta)$. Optimizing over feasible
$K_\Theta$ proves the theorem.

\end{proof}

\begin{corollary}
\label{cor1}
Consider the excess distortion above the infinite-rate remote baseline:
\begin{equation}
\Delta_{\mathrm{rem}}(R)
=
D_e^{(\mathrm{rem})}(R)-D_{\infty}^{\mathrm{rem}}.
\label{eq34}
\end{equation}
Assume $Q\succ 0$ and that $Q$ and $\Sigma_\Theta$ are simultaneously diagonalizable, i.e.,  there exists an orthogonal $U$ such that
\begin{equation}
U^\top \Sigma_\Theta U=\mathrm{diag}(\sigma_{\Theta,1}^2,\dots,\sigma_{\Theta,k}^2),
\label{eq35}
\end{equation}
and
\begin{equation}
U^\top Q\,U=\mathrm{diag}(q_1,\dots,q_k),
\label{eq36}
\end{equation}
with $q_i>0$. Then an optimal posterior covariance in Theorem~\ref{thm2} is diagonal in this basis,
$K_\Theta^\star(R)=U\,\mathrm{diag}(k_1^\star(R),\dots,k_k^\star(R))\,U^\top$, with
\begin{equation}
k_i^\star(R)
=
\min\left\{\sigma_{\Theta,i}^2,\ \frac{\nu^\star(R)}{q_i}\right\},
\label{eq37}
\end{equation}
where $\nu^\star(R)>0$ is chosen so that the rate constraint is active:
\begin{equation}
\sum_{i=1}^k \log\frac{\sigma_{\Theta,i}^2}{k_i^\star(R)}=2R.
\label{eq38}
\end{equation}

In the interior (no-saturation) regime, $k_i^\star(R)<\sigma_{\Theta,i}^2$ for all $i$,
we have
\begin{equation}
K_\Theta^\star(R)=\nu^\star(R)\,Q^{-1},
\label{eq39}
\end{equation}
with
\begin{equation}
\nu^\star(R)=\exp\!\left(\frac{1}{k}\left[\sum_{i=1}^k\log\!\big(q_i\sigma_{\Theta,i}^2\big)-2R\right]\right),
\label{eq40}
\end{equation}
and the excess distortion equals
\begin{equation}
\Delta_{\mathrm{rem}}(R)=\mathrm{tr}\!\left(QK_\Theta^\star(R)\right)=k\,\nu^\star(R).
\label{eq41}
\end{equation}
\end{corollary}

\begin{proof}
\label{appendix3}

For fixed rate $R$, Theorem~\ref{thm2} is equivalent (up to the constant $D_{\infty}^{\mathrm{rem}}$) to
\begin{equation*}
\min_{0\prec K_\Theta\preceq \Sigma_\Theta}\ \mathrm{tr}(QK_\Theta)
\quad\text{s.t.}\quad
\frac12\log\frac{\det\Sigma_\Theta}{\det K_\Theta}\le R.
\end{equation*}
Equivalently, the constraint is $\log\det K_\Theta \ge \log\det\Sigma_\Theta - 2R$.
Under simultaneous diagonalization, write $K_\Theta=U\,\mathrm{diag}(k_i)\,U^\top$.
Then the problem reduces to
\begin{equation*}
\min_{0<k_i\le\sigma_{\Theta,i}^2}\ \sum_{i=1}^k q_i k_i
\quad\text{s.t.}\quad
\sum_{i=1}^k \log k_i \ge \sum_{i=1}^k \log\sigma_{\Theta,i}^2 - 2R.
\end{equation*}
The KKT conditions yield $k_i^\star=\min\{\sigma_{\Theta,i}^2,\nu/q_i\}$ with an active log constraint,
giving \eqref{eq37}--\eqref{eq41}.

\end{proof}

\begin{remark}
\label{remark5}
Corollary~\ref{cor1} assumes $Q\succ 0$ and simultaneous diagonalization with $\Sigma_\Theta$
to express the solution in the scalar parameters $\{q_i,\sigma_{\Theta,i}^2\}$.
Without this commutativity assumption, the same reverse-waterfilling structure holds
in the eigenbasis of the symmetric matrix
\begin{equation}
A_\Theta := \Sigma_\Theta^{1/2}Q\,\Sigma_\Theta^{1/2}.
\label{eq109}
\end{equation}
If $A_\Theta$ has nonpositive eigenvalues, the corresponding normalized posterior
eigenvalues saturate at one (no compression in those directions), and rate is allocated
only across the positive-eigenvalue subspace.
\end{remark}

\subsection{Full-Information Encoder: Rate-Constrained Gaussian Persuasion}

When the encoder observes $(X,\Theta)$, it can also manipulate posterior correlations
between the state $X$ and semantic noise $V$, introducing strategic degrees of freedom. 
%absent in classical rate--distortion theory.
%
%An operational rate-$R$ code constrains the joint informativeness of the message:
%\begin{equation}
%I(X^n,\Theta^n;M)\le nR.
%\label{eq42}
%\end{equation}
%Equivalently, since $\Theta=BX+V$ is invertible in $(X,V)$, the single-letter rate constraint
%can be stated for the augmented Gaussian source $W:=(X,V)$.
Recall $X\sim\mathcal N(0,\Sigma_X)$ and $V\sim\mathcal N(0,\Sigma_V)$ are independent, and 
$\Theta=BX+V$. In the full-information regime the encoder observes $(X,\Theta)$ and sends a
rate-limited message $M$. Define
\begin{equation}
W=
\begin{bmatrix}
X\\
V
\end{bmatrix},
\label{eq43}
\end{equation}
with prior covariance
\begin{equation}
\Sigma_W=\mathrm{blkdiag}(\Sigma_X,\Sigma_V),
\label{eq44}
\end{equation}
and posterior error covariance
\begin{equation}
K_W=\Sigma_{W\mid M}.
\label{eq45}
\end{equation}

\begin{theorem}
\label{thm3}

The strategic rate--distortion function when the encoder observes $X, \Theta$ is
\begin{equation}
R(D_e)
=
\min_{0\preceq K_W \preceq\Sigma_W}
\left\{
\frac12\log\frac{\det\Sigma_W}{\det K_W}
:\;
\mathrm{tr}(\bar W K_W) \le D_e'
\right\},
\label{eq122}
\end{equation}
where $D_e' := D_e - \mathrm{tr}(W_e C_0)$ and 

%\begin{equation}
%0\preceq K_W\preceq \Sigma_W,
%\label{eq46}
%\end{equation}
%together with the entropy budget
%\begin{equation}
%\log\det K_W \ge \log\det\Sigma_W - 2R.
%\label{eq47}
%\end{equation}
%Moreover, the semantic distortion admits the affine representation
%\begin{equation}
%D_e=\mathrm{tr}(W_e C_0)+\mathrm{tr}(\bar W K_W),
%\label{eq48}
%\end{equation}
%where
\begin{equation}
\bar W=\begin{bmatrix}\tilde W & W_e\\ W_e & 0\end{bmatrix}.
\label{eq49}
\end{equation}
and is achievable by a Gaussian test channel of the form
$U=W+Z$ where $Z$ is Gaussian.
\end{theorem}

\begin{proof}

\label{appendix4}

Since $V=\Theta-BX$ is a deterministic invertible linear transform of $(X,\Theta)$, the encoder
effectively observes $(X,V)$, and for any message $M$,
\begin{equation*}
I(X,\Theta;M)=I(X,V;M)=I(W;M).
\end{equation*}

Because $|{\cal M}_n|\le \exp\{nr\}$ we have $H(M)\le nR$ and hence
\begin{equation*}
I(W^n;M)\le H(M)\le nR.
\end{equation*}
We single-letterize as in standard RD converses. Using the chain rule and
conditioning reduces entropy, for i.i.d.\ $W_t$,
\begin{equation*}
I(W^n;M)
=
\sum_{t=1}^n I(W_t;M\mid W^{t-1})
\ge
\sum_{t=1}^n I(W_t;M).
\end{equation*}
Let $Q\sim\mathrm{Unif}\{1,\dots,n\}$ be independent of everything and define
$W:=W_Q$ and $U:=(M,Q)$. Then
\begin{align*}
I(W;U)
&=I(W_Q;M,Q)\\
&=\frac{1}{n}\sum_{t=1}^n I(W_t;M)\\
&\le \frac{1}{n}I(W^n;M)\\
&\le R.
\end{align*}
Apply Lemma~\ref{lemma2} to the Gaussian vector $W$ and auxiliary $U$:
\begin{equation*}
I(W;U)\ge\frac12\log\frac{\det\Sigma_W}{\det \Sigma_{W\mid U}}.
\end{equation*}
Let $K_W:=\Sigma_{W\mid U}$. Since $U=(M,Q)$, this is exactly the per-letter posterior
error covariance induced by the code after time-sharing. Combining yields
\begin{equation*}
R\ge \frac12\log\frac{\det\Sigma_W}{\det K_W},
\end{equation*}
equivalently,
\begin{equation*}
\log\det K_W \ge \log\det\Sigma_W - 2R.
\end{equation*}
Finally, by the law of total covariance, $0\preceq K_W\preceq \Sigma_W$.

\smallskip

Let $\hat X:=\mathbb E[X\mid M]$ and define the MMSE error $e:=X-\hat X$.
Also define $\hat V:=\mathbb E[V\mid M]$ and $v:=V-\hat V$. Then
\begin{equation*}
K_W
=
\Sigma_{W\mid M}
=
\mathrm{Cov}\!\left(\begin{bmatrix}e\\v\end{bmatrix}\right)
=
\begin{bmatrix}
K_X & K_{XV}\\
K_{VX} & K_V
\end{bmatrix},
\end{equation*}
where $K_X=\mathrm{Cov}(e)$ and $K_{XV}=\mathrm{Cov}(e,v)$.

By the orthogonality principle, $e$ is orthogonal to any measurable function of $M$, hence
$e\perp \hat V$ and therefore
\begin{equation*}
\mathrm{Cov}(e,V)=\mathrm{Cov}(e,\hat V+v)=\mathrm{Cov}(e,v)=K_{XV}.
\end{equation*}
Since $\mathrm{Cov}(X,V)=0$ by independence, we also have
\begin{align*}
0=\mathrm{Cov}(X,V)
=\mathrm{Cov}(\hat X+e,V)
=\mathrm{Cov}(\hat X,V)+\mathrm{Cov}(e,V),
\end{align*}
which yields
\begin{equation*}
\mathrm{Cov}(\hat X,V)=-K_{XV}.
\end{equation*}

Next,
\begin{equation*}
\Theta-\hat X = BX+V-\hat X = (B-I)\hat X + Be + V.
\end{equation*}
Using $e\perp \hat X$, the covariance is
\begin{align*}
\mathrm{Cov}(\Theta-\hat X)
&=(B-I)\mathrm{Cov}(\hat X)(B-I)^\top + B K_X B^\top + \Sigma_V \\
&\quad + (B-I)\mathrm{Cov}(\hat X,V)+\mathrm{Cov}(V,\hat X)(B-I)^\top \\
&\quad + B\mathrm{Cov}(e,V)+\mathrm{Cov}(V,e)B^\top.
\end{align*}
Using $\mathrm{Cov}(\hat X)=\Sigma_X-K_X$, $\mathrm{Cov}(\hat X,V)=-K_{XV}$, and
$\mathrm{Cov}(e,V)=K_{XV}$, the cross terms simplify to
\begin{equation*}
-(B-I)K_{XV}-K_{VX}(B-I)^\top + BK_{XV}+K_{VX}B^\top = K_{XV}+K_{VX}.
\end{equation*}
The $K_X$-terms simplify as in Lemma~\ref{lemma1}:
\begin{align*}
(B-I)(\Sigma_X-K_X)(B-I)^\top + BK_XB^\top
&=
(B-I)\Sigma_X(B-I)^\top + \big(BK_X+K_XB^\top-K_X\big).
\end{align*}
Hence
\begin{equation*}
\mathrm{Cov}(\Theta-\hat X)
=
C_0 + \big(BK_X+K_XB^\top-K_X\big) + (K_{XV}+K_{VX}).
\end{equation*}
Therefore,
\begin{align*}
D_e
&=\mathrm{tr}\!\left(W_e\,\mathrm{Cov}(\Theta-\hat X)\right)\\
&=\mathrm{tr}(W_eC_0)
  +\mathrm{tr}\!\left((W_eB+B^\top W_e-W_e)K_X\right)
  +\mathrm{tr}\!\left(W_e(K_{XV}+K_{VX})\right).
\end{align*}
Defining $\tilde W:=W_eB+B^\top W_e-W_e$ and
\begin{equation*}
\bar W:=\begin{bmatrix}\tilde W & W_e\\ W_e & 0\end{bmatrix}, 
\end{equation*}
we have:
\begin{equation*}
D_e
=
\mathrm{tr}(W_eC_0)+\mathrm{tr}(\bar W K_W). 
\end{equation*}

\smallskip

\emph{Converse:} The preceding analysis shows that any achievable $(R,D_e)$ induces a matrix $K_W$
satisfying $0\preceq K_W\preceq\Sigma_W$ and $\log\det K_W\ge\log\det\Sigma_W-2R$, and achieves
$D_e=\mathrm{tr}(W_eC_0)+\mathrm{tr}(\bar W K_W)$.

\emph{Achievability:} Let $Q\sim\mathrm{Unif}\{1,\dots,n\}$ independent of everything, define $W:=W_Q$ and $U:=(M,Q)$.
Then
\[
I(W;U)=I(W_Q;M,Q)\le I(W^n;M)\le H(M)\le nR,
\]
hence $I(W;U)\le R$.
The decoder
reconstructs $U^n$ and forms $\hat X^n=\mathbb E[X^n\mid U^n]$, which induces the desired posterior
error covariance $K_W$ and achieves distortion $\mathrm{tr}(W_eC_0)+\mathrm{tr}(\bar W K_W)$.
Optimizing over feasible $K_W$ proves the theorem. 
\end{proof}

When $\Sigma_X,\Sigma_V,B$, and $W_e$ share an eigenbasis, the covariance-design program from
Theorem~\ref{thm3} decouples across eigen-directions. Each direction reduces to a two-dimensional
posterior-design problem for the pair $(x_i,v_i)$, with a determinant budget allocated by a
scalar waterfilling law.

\begin{assumption}
\label{assumption1}
There exists an orthogonal matrix $U$ such that
\begin{align}
U^\top\Sigma_XU&=\mathrm{diag}(\sigma_1^2,\dots,\sigma_k^2),
\label{eq50}\\
U^\top\Sigma_VU&=\mathrm{diag}(\tau_1^2,\dots,\tau_k^2),
\label{eq51}\\
U^\top BU&=\mathrm{diag}(b_1,\dots,b_k),
\label{eq52}\\
U^\top W_eU&=\mathrm{diag}(w_1,\dots,w_k),
\label{eq53}
\end{align}
with $\sigma_i>0$, $\tau_i>0$, and $w_i>0$ for all $i$.
\end{assumption}
\begin{remark}
The condition $\tau_i>0$ ensures $\log\det\Sigma_V$ is well-defined and guarantees that
each $2\times 2$ block in Theorem~\ref{thm4} has one positive and one negative eigenvalue.
\end{remark}

\begin{theorem}
\label{thm4}
Under Assumption~\ref{assumption1}, there exists an optimal posterior error covariance
$K_W^\star(R)$ that is block-diagonal across $i$ in the basis
$\bar U:=\mathrm{blkdiag}(U,U)$:
\begin{equation}
K_W^\star(R)=\bar U\ \mathrm{blkdiag}\!\big(K_1^\star(R),\dots,K_k^\star(R)\big)\ \bar U^\top,
\label{eq54}
\end{equation}
with $2\times2$ blocks
\begin{equation}
K_i^\star(R)=D_i\,\tilde K_i^\star(R)\,D_i,
\label{eq55}
\end{equation}
where
\begin{equation}
D_i=\mathrm{diag}(\sigma_i,\tau_i).
\label{eq56}
\end{equation}
Each normalized block can be written as
\begin{equation}
\tilde K_i^\star(R)=R_i\,\mathrm{diag}(\delta_i^\star(R),1)\,R_i^\top,
\label{eq57}
\end{equation}
where $0<\delta_i^\star(R)\le 1$ and $R_i$ diagonalizes
\begin{equation}
A_i=
\begin{bmatrix}
\tilde w_i\,\sigma_i^2 & w_i\,\sigma_i\tau_i\\
w_i\,\sigma_i\tau_i & 0
\end{bmatrix},
\label{eq58}
\end{equation}
with
\begin{equation}
\tilde w_i=w_i(2b_i-1).
\label{eq59}
\end{equation}
Writing $R_i^\top A_i R_i=\mathrm{diag}(a_i^+,a_i^-)$ with $a_i^+>0>a_i^-$, we have
\begin{equation}
a_i^+=\frac12\left(\tilde w_i\sigma_i^2+\sqrt{\tilde w_i^2\sigma_i^4+4w_i^2\sigma_i^2\tau_i^2}\right).
\label{eq60}
\end{equation}

The scalars $\delta_i^\star(R)$ satisfy the waterfilling law
\begin{equation}
\delta_i^\star(R)=\min\!\left\{1,\ \frac{\eta^\star(R)}{a_i^+}\right\},
\label{eq61}
\end{equation}
where $\eta^\star(R)>0$ is chosen so that the normalized log-det constraint is active:
\begin{equation}
\sum_{i=1}^k \log \delta_i^\star(R)=-2R.
\label{eq62}
\end{equation}
\end{theorem}

\begin{proof}
\label{appendix5}

Define $\bar U:=\mathrm{blkdiag}(U,U)$ and transform
\begin{align*}
K'&=\bar U^\top K_W\bar U,\\
\Sigma'&=\bar U^\top \Sigma_W\bar U.
\end{align*}
Under Assumption~\ref{assumption1}, we have
\begin{equation*}
\Sigma'=\mathrm{blkdiag}(\Sigma_X',\Sigma_V')
=
\mathrm{blkdiag}\!\left(\mathrm{diag}(\sigma_i^2),\mathrm{diag}(\tau_i^2)\right).
\end{equation*}
Because trace and determinant are invariant under orthogonal transforms, the covariance-design
problem in Theorem~\ref{thm3} is equivalent to
\begin{equation*}
\min_{0\preceq K'\preceq \Sigma'}\ \mathrm{tr}(\bar W' K')
\quad\text{s.t.}\quad
\log\det K' \ge \log\det\Sigma' - 2R,
\end{equation*}
where $\bar W':=\bar U^\top \bar W\,\bar U$ is block diagonal across $i$ with $2\times 2$ blocks
\begin{equation*}
\bar W_i'=\begin{bmatrix}\tilde w_i & w_i\\ w_i & 0\end{bmatrix},
\qquad
\tilde w_i:=w_i(2b_i-1).
\end{equation*}

Normalize by $\Sigma'$: define $\tilde K := \Sigma'^{-1/2}K'\Sigma'^{-1/2}$. Then
$0\preceq K'\preceq \Sigma'$ becomes $0\preceq \tilde K\preceq I$, and
\begin{align*}
\log\det K'
&=\log\det\Sigma' + \log\det\tilde K,\\
\mathrm{tr}(\bar W'K')
&=\mathrm{tr}\!\left(\Sigma'^{1/2}\bar W'\Sigma'^{1/2}\ \tilde K\right).
\end{align*}
Let
\begin{equation*}
A:=\Sigma'^{1/2}\bar W'\Sigma'^{1/2},
\end{equation*}
which is block diagonal across $i$ with blocks
\begin{equation*}
A_i=
\begin{bmatrix}
\tilde w_i\,\sigma_i^2 & w_i\,\sigma_i\tau_i\\
w_i\,\sigma_i\tau_i & 0
\end{bmatrix}.
\end{equation*}
Thus the normalized problem is
\begin{equation*}
\min_{0\preceq \tilde K\preceq I}\ \mathrm{tr}(A\tilde K)
\quad\text{s.t.}\quad
\log\det\tilde K \ge -2R.
\end{equation*}

We next reduce it to $k$ independent $2\times 2$ blocks. Partition $\tilde K$ into $k\times k$
blocks of size $2\times 2$:
\begin{equation*}
\tilde K=[\tilde K_{ij}]_{i,j=1}^k.
\end{equation*}
Since $A$ is block diagonal,
\begin{equation*}
\mathrm{tr}(A\tilde K)=\sum_{i=1}^k \mathrm{tr}(A_i\tilde K_{ii}),
\end{equation*}
so off-diagonal blocks do not affect the objective.

For feasibility, note that $\tilde K\succeq 0$ implies the determinant inequality
\begin{equation*}
\det(\tilde K)\le \prod_{i=1}^k \det(\tilde K_{ii}).
\end{equation*}
Define the block-diagonal matrix $\tilde K^{\mathrm{bd}}:=\mathrm{blkdiag}(\tilde K_{11},\dots,\tilde K_{kk})$.
Then $\tilde K^{\mathrm{bd}}\preceq I$ and
\begin{equation*}
\det(\tilde K^{\mathrm{bd}})=\prod_{i=1}^k \det(\tilde K_{ii})\ge \det(\tilde K).
\end{equation*}
Hence $\tilde K^{\mathrm{bd}}$ is feasible whenever $\tilde K$ is feasible, and it achieves
the same objective value. Therefore, there exists an optimal solution that is block diagonal:
\begin{equation*}
\tilde K^\star=\mathrm{blkdiag}(\tilde K_1^\star,\dots,\tilde K_k^\star),
\qquad
0\preceq \tilde K_i^\star\preceq I.
\end{equation*}
The constraint becomes
\begin{equation*}
\log\det\tilde K^\star=\sum_{i=1}^k \log\det(\tilde K_i^\star)\ge -2R.
\end{equation*}

\smallskip

Fix $i$ and take the eigen-decomposition
\begin{equation*}
A_i=R_i\,\mathrm{diag}(a_i^+,a_i^-)\,R_i^\top,
\qquad
a_i^+>0>a_i^-.
\end{equation*}
The eigenvalues are
\begin{equation*}
a_i^\pm
=
\frac{1}{2}\left(
\tilde w_i\sigma_i^2
\pm
\sqrt{\tilde w_i^2\sigma_i^4 + 4w_i^2\sigma_i^2\tau_i^2}
\right),
\end{equation*}
so indeed $a_i^+a_i^-=\det(A_i)=-w_i^2\sigma_i^2\tau_i^2<0$.

Let $\tilde K_i$ be any feasible $2\times2$ PSD matrix with eigenvalues
$0<\kappa_{i,1}\le \kappa_{i,2}\le 1$. For fixed eigenvalues $(\kappa_{i,1},\kappa_{i,2})$,
the minimum of $\mathrm{tr}(A_i\tilde K_i)$ is attained by aligning the larger eigenvalue
$\kappa_{i,2}$ with the negative eigenvalue $a_i^-$. Thus we may restrict attention to
\begin{equation*}
\tilde K_i=R_i\,\mathrm{diag}(\kappa_{i,1},\kappa_{i,2})\,R_i^\top,
\end{equation*}
which yields cost $a_i^+\kappa_{i,1}+a_i^-\kappa_{i,2}$.

Since $a_i^-<0$, increasing $\kappa_{i,2}$ strictly decreases the objective and increases the
determinant $\det(\tilde K_i)=\kappa_{i,1}\kappa_{i,2}$, which preserves feasibility under a
lower bound on the determinant. Therefore at any optimum we must have $\kappa_{i,2}=1$.
Denote
\begin{equation*}
\delta_i:=\kappa_{i,1}=\det(\tilde K_i)\in(0,1].
\end{equation*}
Hence an optimal block has the form
\begin{equation*}
\tilde K_i^\star = R_i\,\mathrm{diag}(\delta_i^\star,1)\,R_i^\top,
\qquad \delta_i^\star\in(0,1],
\end{equation*}
and contributes cost $\mathrm{tr}(A_i\tilde K_i^\star)=a_i^- + a_i^+\delta_i^\star$.

Collecting blocks yields
\begin{align*}
\mathrm{tr}(A\tilde K^\star)
&=\sum_{i=1}^k \big(a_i^- + a_i^+\delta_i\big),\\
\log\det(\tilde K^\star)
&=\sum_{i=1}^k \log\delta_i.
\end{align*}
Since $\sum_i a_i^-$ is constant, the normalized design reduces to the scalar convex program
\begin{equation*}
\min_{0<\delta_i\le 1}\ \sum_{i=1}^k a_i^+\delta_i
\quad\text{s.t.}\quad
\sum_{i=1}^k \log\delta_i \ge -2R.
\end{equation*}
The constraint must be active: if $\sum_i\log\delta_i>-2R$, one can decrease some $\delta_i$
slightly (while keeping $\delta_i\le 1$) to strictly lower the objective.

The KKT conditions yield the waterfilling solution
\begin{equation*}
\delta_i^\star(R)=\min\left\{1,\ \frac{\eta^\star(R)}{a_i^+}\right\},
\end{equation*}
where $\eta^\star(R)>0$ is chosen so that the (active) constraint holds:
\begin{equation*}
\sum_{i=1}^k \log\delta_i^\star(R)=-2R.
\end{equation*}
Undoing the normalization yields the stated form \eqref{eq54}--\eqref{eq62}.
\end{proof}

%%%%%%%%%%%%%%%%%%%%%%%%%%%%%%%%%%%%%%%%%%%%%%%%%%%%
\section{Implications in Learning Systems}
\label{section5}
%%%%%%%%%%%%%%%%%%%%%%%%%%%%%%%%%%%%%%
\subsection{Multimodal Encoders and Elimination of the Remote Penalty}
\label{subsection5A}

We study a remote sensing regime where the encoder observes noisy modalities of the latent
state $X$ (e.g., vision/audio/text features) that are conditionally independent given $X$.

Let $X\sim\mathcal N(0,\Sigma_X)$ with $\Sigma_X\succ 0$, and suppose the encoder observes $m$ conditionally independent
modalities
\begin{equation}
Z_j = H_j X + W_j,
\label{eq63}
\end{equation}
for $j=1,\dots,m$, where $W_j\sim\mathcal N(0,R_j)$ are independent of each other and of $(X,V)$.
Let $Z^{(m)}:=(Z_1,\dots,Z_m)$, and define the MMSE proxy
\begin{equation}
S^{(m)} := \mathbb E[X\mid Z^{(m)}].
\label{eq64}
\end{equation}
Define
\begin{equation}
\Sigma_S^{(m)} := \mathrm{Cov}(S^{(m)}).
\label{eq65}
\end{equation}
Define the cumulative precision
\begin{equation}
J_m := \sum_{j=1}^m H_j^\top R_j^{-1}H_j \succeq 0.
\label{eqJm}
\end{equation}

Define the normalized recoverability matrix
\begin{equation}
\Gamma_m := \Sigma_X^{-1/2}\,\Sigma_S^{(m)}\,\Sigma_X^{-1/2}.
\label{eq66}
\end{equation}
Equivalently,
\begin{equation}
\Gamma_m = I - \Sigma_X^{-1/2}\Sigma_{X\mid Z^{(m)}}\Sigma_X^{-1/2}.
\label{eq67}
\end{equation}
Define the geometric-mean recoverability factor
\begin{equation}
G_m := \left(\frac{\det \Sigma_S^{(m)}}{\det \Sigma_X}\right)^{1/k}.
\label{eq68}
\end{equation}
Equivalently,
\begin{equation}
G_m = \left(\prod_{i=1}^k \gamma_i^{(m)}\right)^{1/k},
\label{eq69}
\end{equation}
where $\{\gamma_i^{(m)}\}$ are the eigenvalues of $\Gamma_m$.

\begin{theorem}
\label{thm5}
The posterior covariance and recoverable covariance satisfy
\begin{align}
\Sigma_{X\mid Z^{(m)}}^{-1}
&=
\Sigma_X^{-1}+J_m,
\label{eq71b}\\
\Sigma_S^{(m)}
&=
\Sigma_X-\Sigma_{X\mid Z^{(m)}}.
\label{eq70b}
\end{align}
Moreover, $\Sigma_{X\mid Z^{(m)}}$ is monotone decreasing in $m$ in Loewner order,
and $\Sigma_S^{(m)}$ is monotone increasing.

If $J_m\succ 0$ (equivalently $\mathrm{Null}(J_m)=\{0\}$), then
\begin{equation}
0 \prec \Sigma_{X\mid Z^{(m)}} \prec \Sigma_X
\quad\text{and}\quad
0\prec \Sigma_S^{(m)} \prec \Sigma_X,
\label{eqSpanFinite}
\end{equation}
hence $G_m\in(0,1)$.

If in addition $\lambda_{\min}(J_m)\to\infty$ as $m\to\infty$, then
\begin{equation}
\Sigma_{X\mid Z^{(m)}} \to 0,
\qquad
\Sigma_S^{(m)}\to \Sigma_X,
\qquad
G_m\to 1.
\label{eqSpanAsymp}
\end{equation}
\end{theorem}

\begin{remark}
\label{remark:spanning}

The cumulative precision matrix $J_m$
plays a central role in linear Gaussian inference. In vector {remote} and {indirect} Gaussian source coding, it is standard to assume that
the measurement operator is full rank (or that the induced precision is nonsingular), since otherwise
some components of the source are never observed and the indirect coding penalty cannot vanish \cite{tian2009remote,rini2019compress}.
In linear filtering and sensor fusion \cite{bocquet2017degenerate, plarre2009kalman,sinopoli2004kalman}, $J_m$ coincides with the accumulated measurement precision
mapped into the state space, and conditions such as $J_m \succ 0$ or $\lambda_{\min}(J_m)\to\infty$
are the static counterparts of observability and persistent excitation assumptions that ensure
error covariance contraction.
Similarly, in Bayesian linear inverse problems, it is well known that posterior contraction occurs
only on the \emph{likelihood-informed subspace}, while directions orthogonal to the range of
$J_m$ retain their prior variance \cite{carere2024optimal}.

In the present semantic compression setting, the condition $J_m\succ 0$ ensures that no direction
of the latent state $X$ is permanently hidden from the encoder, while the stronger condition
$\lambda_{\min}(J_m)\to\infty$ guarantees that the posterior covariance $\Sigma_{X\mid Z^{(m)}}$
collapses to zero as the number of modalities grows.
Without these spanning conditions, there exists a nontrivial subspace along which the encoder
cannot improve its estimate regardless of rate or compute, and the semantic performance gap
relative to the direct-encoding benchmark cannot vanish.
\end{remark}

\begin{proof}
Equations \eqref{eq71b}--\eqref{eq70b} are standard linear-Gaussian MMSE identities.

Monotonicity follows since $J_{m+1}=J_m+H_{m+1}^\top R_{m+1}^{-1}H_{m+1}\succeq J_m$,
hence $\Sigma_{X\mid Z^{(m+1)}}^{-1}\succeq \Sigma_{X\mid Z^{(m)}}^{-1}$ and therefore
$\Sigma_{X\mid Z^{(m+1)}}\preceq \Sigma_{X\mid Z^{(m)}}$.

If $J_m\succ 0$, then $\Sigma_X^{-1}+J_m\succ \Sigma_X^{-1}$, so
$(\Sigma_X^{-1}+J_m)^{-1} \prec \Sigma_X$, yielding \eqref{eqSpanFinite}.

If $\lambda_{\min}(J_m)\to\infty$, then $\lambda_{\min}(\Sigma_X^{-1}+J_m)\to\infty$,
so $\|\Sigma_{X\mid Z^{(m)}}\|=\|(\Sigma_X^{-1}+J_m)^{-1}\|\to 0$, proving \eqref{eqSpanAsymp}.
The statements for $\Sigma_S^{(m)}$ and $G_m$ follow from \eqref{eq70b} and continuity of $\det(\cdot)$
on $\mathbb S_{++}^k$.
\end{proof}

%\begin{theorem}
%
%The recoverable covariance satisfies
%\begin{equation}
%\Sigma_S^{(m)}=\Sigma_X-\Sigma_{X\mid Z^{(m)}},
%\label{eq70}
%\end{equation}
%with
%\begin{equation}
%\Sigma_{X\mid Z^{(m)}}^{-1}
%=
%\Sigma_X^{-1} + \sum_{j=1}^m H_j^\top R_j^{-1} H_j.
%\label{eq71}
%\end{equation}
%Moreover, $\Sigma_S^{(m)}$ is monotone increasing in $m$ (wrt. Loewner order), hence
%$\Gamma_m$ and $G_m$ are nondecreasing. If the cumulative precision diverges in all
%directions, then $\Sigma_{X\mid Z^{(m)}}\to 0$, $\Sigma_S^{(m)}\to\Sigma_X$, and $G_m\to 1$.
%\end{theorem}
%
%\begin{proof}
%Equations~\eqref{eq70} and~\eqref{eq71} are standard for linear-Gaussian estimation.
%Monotonicity follows since adding a modality increases the precision matrix, which decreases
%$\Sigma_{X\mid Z^{(m)}}$ and thus increases $\Sigma_S^{(m)}=\Sigma_X-\Sigma_{X\mid Z^{(m)}}$.
%The remaining claims follow by congruence with $\Sigma_X^{-1/2}$ and the divergence condition.
%\end{proof}

\begin{corollary}
\label{cor2}
Let $D_e^{(m)}(R)$ denote the optimal semantic distortion at rate $R$ when the encoder
observes $Z^{(m)}=(Z_1,\dots,Z_m)$ with $Z_j=H_jX+W_j$ as in Theorem~\ref{thm5},
and the semantic variable is $\Theta=BX+V$ with $V\perp (X,Z^{(m)})$.
Let $D_e^{(\mathrm{dir})}(R)$ denote the corresponding distortion when the encoder observes $X$.

Then:
\begin{enumerate}
\item For every $m$,
\begin{equation}
D_e^{(m)}(R)
=
\mathrm{tr}(W_eC_0)+\mathrm{tr}\!\big(\tilde W\,\Sigma_{X\mid Z^{(m)}}\big)+\Delta_m(R),
\label{eq72b}
\end{equation}
where
\begin{equation}
\Delta_m(R)
:=
\min_{0\prec K \preceq \Sigma_S^{(m)}}
\left\{
\mathrm{tr}(\tilde W K):
\frac12\log\frac{\det\Sigma_S^{(m)}}{\det K}\le R
\right\}.
\label{eq73b}
\end{equation}

\item If $\lambda_{\min}(J_m)\to\infty$ as $m\to\infty$ (equivalently $\Sigma_{X\mid Z^{(m)}}\to 0$),
then for every fixed $R$,
\begin{equation}
\lim_{m\to\infty} D_e^{(m)}(R)=D_e^{(\mathrm{dir})}(R).
\label{eq76b}
\end{equation}

\item If instead $\mathrm{Null}(J_\infty)\neq\{0\}$ where $J_\infty:=\sum_{j\ge 1} H_j^\top R_j^{-1}H_j$
(convergent in the PSD sense), then $\Sigma_{X\mid Z^{(m)}}$ does not converge to $0$ and there remains
a nonzero gap to the direct benchmark for any fixed rate $R$.
\end{enumerate}
\end{corollary}

\begin{proof}
Since $V\perp (X,Z^{(m)})$ and $M=f(Z^{(m)})$, we have $V\perp M$, hence
$D_e=\mathrm{tr}(W_eC_0)+\mathrm{tr}(\tilde W K_X)$ with $K_X=\Sigma_{X\mid M}$.

For Gaussian $Z^{(m)}$, write the MMSE decomposition $X=S^{(m)}+N^{(m)}$ with
$S^{(m)}=\mathbb E[X\mid Z^{(m)}]$ and $N^{(m)}\perp Z^{(m)}$. For any $M=f(Z^{(m)})$,
$\hat X=\mathbb E[X\mid M]=\mathbb E[S^{(m)}\mid M]$, hence by the law of total covariance,
\[
\Sigma_{X\mid M}=\Sigma_{X\mid Z^{(m)}}+\Sigma_{S^{(m)}\mid M}.
\]
Let $K:=\Sigma_{S^{(m)}\mid M}$. The rate constraint implies $R\ge I(S^{(m)};M)$, and
Lemma~\ref{lemma2} gives $I(S^{(m)};M)\ge \frac12\log\frac{\det\Sigma_S^{(m)}}{\det K}$,
yielding \eqref{eq73b}. Substituting into $D_e$ yields \eqref{eq72b}.

If $\lambda_{\min}(J_m)\to\infty$, then by Theorem~\ref{thm5} we have
$\Sigma_{X\mid Z^{(m)}}\to 0$ and $\Sigma_S^{(m)}\to \Sigma_X$.
The optimization defining $\Delta_m(R)$ is continuous under this limit (the feasible sets
converge and the objective is linear), hence $\Delta_m(R)\to \Delta_{\mathrm{dir}}(R)$ and
\eqref{eq76b} follows.
If $\mathrm{Null}(J_\infty)\neq\{0\}$, then $\Sigma_{X\mid Z^{(m)}}$ retains a nonzero component
on the unobserved subspace, so the gap to the direct benchmark cannot vanish.
\end{proof}

\begin{remark}
Theorem~\ref{thm2} addresses the case where the encoder observes the semantic proxy
$\Theta=BX+V$, so $M$ can depend on $V$ and the encoder's objective involves an additional
$X$--$V$ cross term. In contrast, the multimodal model in \eqref{eq63} depends only on $X$
(plus sensor noise), so $V\perp M$ and the problem reduces to indirect Gaussian coding of $X$.
Accordingly, multimodality closes the gap to the direct-$X$ benchmark but cannot recover
independent semantic randomness $V$ that is never observed.
\end{remark}

%%%%%%%%%%%%%%%%%%%%%%%%%%%%%%%%%%%%%%%%%%%%
\subsection{Compute as an Information Budget and Scaling of Semantic Precision}
\label{subsection5B}

This subsection connects architectural and inference-time compute to an information budget
governing how sharply a model can refine its posterior over a latent state. The guiding idea
is that each layer (or inference step) can extract only a limited amount of new information
from the input (or memory) before passing a refined belief onward.

Let $X\sim\mathcal N(0,\Sigma_X)$ be the latent state and let $Z$ denote the available data
correlated with $X$. A representation learner produces a sequence of internal states
$H^{(0)},H^{(1)},\dots,H^{(L)}$ via a (possibly randomized) update rule
\begin{equation}
H^{(\ell+1)} \sim p_\ell(\cdot\mid H^{(\ell)},Z),
\label{eq77}
\end{equation}
for $\ell=0,\dots,L-1$, so that the model can re-access $Z$ at each step.

We interpret limited compute at step $\ell$ as a bound on the incremental information extracted from $Z$:
\begin{equation}
I\!\big(Z;H^{(\ell+1)}\mid H^{(\ell)}\big)\le R_\ell,
\label{eq78}
\end{equation}
for $\ell=0,\dots,L-1$.

Since $H^{(\ell+1)}$ is generated from $(H^{(\ell)},Z)$, we have the Markov property
$X-(H^{(\ell)},Z)-H^{(\ell+1)}$, hence by data processing
\begin{equation}
I\!\big(X;H^{(\ell+1)}\mid H^{(\ell)}\big)\le R_\ell,
\label{eq79}
\end{equation}
for $\ell=0,\dots,L-1$.

Define the posterior error covariance after step $\ell$ by
\begin{equation}
K_X^{(\ell)}:=\Sigma_{X\mid H^{(\ell)}}.
\label{eq80}
\end{equation}
We take
\begin{equation}
K_X^{(0)}=\Sigma_X.
\label{eq81}
\end{equation}

\begin{theorem}
\label{thm6}
Let $X\sim\mathcal N(0,\Sigma_X)$ and let $H^{(0)},\dots,H^{(L)}$ be any sequence of
representations generated from $(H^{(\ell)},Z)$. Assume that for each $\ell$,
$X-(H^{(\ell)},Z)-H^{(\ell+1)}$ and \eqref{eq78} holds. Then
\begin{equation}
\log\det K_X^{(L)} \ge \log\det\Sigma_X - 2\sum_{\ell=0}^{L-1}R_\ell.
\label{eq82}
\end{equation}
\end{theorem}

\begin{proof}
By the chain rule,
\begin{align*}
I(X;H^{(L)})
&\leq 
I(X;H^{(0)}) + \sum_{\ell=0}^{L-1} I\!\big(X;H^{(\ell+1)}\mid H^{(\ell)}\big).
\end{align*}
By data processing under $X-(H^{(\ell)},Z)-H^{(\ell+1)}$ and the incremental budget \eqref{eq78},
we have
\begin{align*}
I\!\big(X;H^{(\ell+1)}\mid H^{(\ell)}\big)
&\le
I\!\big(Z;H^{(\ell+1)}\mid H^{(\ell)}\big)\\
&\le R_\ell.
\end{align*}
Therefore,
\begin{equation*}
I(X;H^{(L)}) \le I(X;H^{(0)}) + \sum_{\ell=0}^{L-1} R_\ell.
\end{equation*}
Since $X$ is Gaussian, Lemma~\ref{lemma2} implies
\begin{equation*}
I(X;H^{(L)}) \ge \frac12\log\frac{\det\Sigma_X}{\det K_X^{(L)}}.
\end{equation*}
Combining yields \eqref{eq82}.
\end{proof}

\begin{remark}
Constraint \eqref{eq78} should be read as a modeling abstraction: compute does not create
new information about $X$ beyond what is already contained in $Z$, but limited compute and
architecture can prevent the model from extracting and representing that information.
Allowing each step to depend on $(H^{(\ell)},Z)$ captures the empirically important fact
that modern architectures repeatedly re-attend to the input and context.
\end{remark}

Define the total compute budget
\begin{equation}
R_{\mathrm{tot}} := \sum_{\ell=0}^{L-1} R_\ell.
\label{eq83}
\end{equation}
Theorem~\ref{thm6} implies that $R_{\mathrm{tot}}$ acts like an information-rate budget in the
direct strategic RD problem, since any final posterior $K_X^{(L)}$ must satisfy
$\frac12\log\frac{\det\Sigma_X}{\det K_X^{(L)}} \le R_{\mathrm{tot}}$.

Accordingly, define
\begin{equation}
D_{\mathrm{opt}}(R_{\mathrm{tot}})
=
\min_{0\prec K\preceq \Sigma_X}
\left\{
\mathrm{tr}(\tilde W K):
\frac12\log\frac{\det\Sigma_X}{\det K}\le R_{\mathrm{tot}}
\right\}.
\label{eq84}
\end{equation}

\begin{theorem}
\label{thm7}
Let $R_{\mathrm{tot}}$ be defined by \eqref{eq83}. Then any mechanism satisfying \eqref{eq79}
must obey
\begin{equation}
\mathrm{tr}(\tilde W K_X^{(L)})\ge D_{\mathrm{opt}}(R_{\mathrm{tot}}),
\label{eq85}
\end{equation}
and this lower bound is achievable by a Gaussian test-channel.

Moreover, if the optimizer satisfies the interior condition $K^\star\prec \Sigma_X$,
then the unique optimum is
\begin{equation}
K^\star(R_{\mathrm{tot}})=\nu^\star(R_{\mathrm{tot}})\,\tilde W^{-1},
\label{eq86}
\end{equation}
with
\begin{equation}
\nu^\star(R_{\mathrm{tot}})
=
\exp\!\left(\frac{1}{k}\big[\log\det\Sigma_X+\log\det\tilde W-2R_{\mathrm{tot}}\big]\right),
\label{eq87}
\end{equation}
and the optimal semantic distortion decays exponentially:
\begin{equation}
D_{\mathrm{opt}}(R_{\mathrm{tot}})
=
k\,(\det\Sigma_X\,\det\tilde W)^{1/k}\,
\exp\!\left(-2\frac{R_{\mathrm{tot}}}{k}\right).
\label{eq88}
\end{equation}
\end{theorem}

\begin{proof}
The feasibility implication $\frac12\log\frac{\det\Sigma_X}{\det K_X^{(L)}}\le R_{\mathrm{tot}}$
follows directly from Theorem~\ref{thm6}. Since $D_{\mathrm{opt}}(R_{\mathrm{tot}})$ is the minimum
of $\mathrm{tr}(\tilde W K)$ over this feasible set, \eqref{eq85} follows.

Achievability follows from the same Gaussian test-channel construction used in the direct RD
proof: for any feasible $K$ one can construct a Gaussian auxiliary $U=X+Z$ inducing
$\Sigma_{X\mid U}=K$ and $I(X;U)=\frac12\log\frac{\det\Sigma_X}{\det K}\le R_{\mathrm{tot}}$.

For the interior solution, form the Lagrangian
\begin{equation*}
\mathcal L(K,\nu)=\mathrm{tr}(\tilde W K)+\nu\!\left(\frac12\log\frac{\det\Sigma_X}{\det K}-R_{\mathrm{tot}}\right),
\end{equation*}
with $\nu\ge 0$. Stationarity gives $\tilde W - \frac{\nu}{2}K^{-1}=0$, hence
$K=\nu\,\tilde W^{-1}$ after re-scaling the multiplier. Enforcing the active information
constraint yields \eqref{eq87}, and substituting into $\mathrm{tr}(\tilde W K^\star)$ yields
\eqref{eq88}.
\end{proof}
\begin{remark}
\label{remark6}
The interior form \eqref{eq86}--\eqref{eq88} requires $\tilde W\succ 0$ and that the upper bound
$K\preceq\Sigma_X$ is inactive. A convenient sufficient and necessary check can be written using
\begin{equation}
A_X:=\Sigma_X^{1/2}\tilde W\,\Sigma_X^{1/2}.
\label{eq110}
\end{equation}
When \eqref{eq86} holds, the multiplier $\nu^\star(R_{\mathrm{tot}})$ from \eqref{eq87}
must satisfy
\begin{equation}
\nu^\star(R_{\mathrm{tot}}) < \lambda_{\min}(A_X),
\label{eq111}
\end{equation}
which is equivalent to $K^\star(R_{\mathrm{tot}})\prec\Sigma_X$.

If \eqref{eq111} fails, then some directions saturate at the prior, and the optimizer takes the
reverse-waterfilling form in the eigenbasis of $A_X$:
there exists $U$ with $A_X=U\,\mathrm{diag}(\lambda_i)\,U^\top$ such that
\begin{equation}
K^\star(R_{\mathrm{tot}})
=
\Sigma_X^{1/2}U\,\mathrm{diag}(d_i^\star)\,U^\top\Sigma_X^{1/2},
\label{eq112}
\end{equation}
where for $\lambda_i>0$
\begin{equation}
d_i^\star=\min\left\{1,\ \frac{\nu^\star}{\lambda_i}\right\},
\label{eq113}
\end{equation}
and $d_i^\star=1$ whenever $\lambda_i\le 0$.
The scalar $\nu^\star$ is chosen so that the rate constraint is active:
\begin{equation}
\sum_{\lambda_i>0}\log\left(\frac{1}{d_i^\star}\right)=2R_{\mathrm{tot}}.
\label{eq114}
\end{equation}
In this regime the exponential scaling \eqref{eq88} becomes piecewise (each time a mode
desaturates, the effective slope changes).
\end{remark}

\begin{corollary}
\label{cor3}
If $R_\ell=R_0$ for all $\ell$, then in the interior regime,
\begin{equation}
D_{\mathrm{opt}}(L)=k(\det\Sigma_X\,\det\tilde W)^{1/k}\exp\!\left(-2L\frac{R_0}{k}\right).
\label{eq89}
\end{equation}
If an inference procedure produces intermediate states $M_1,\dots,M_T$ satisfying an incremental
budget $I(X;M_t\mid M_{1:t-1})\le R_{\mathrm{step}}$, then with $R_{\mathrm{tot}}=TR_{\mathrm{step}}$
the same exponential law \eqref{eq88} holds with $L$ replaced by $T$.
\end{corollary}

\begin{proof}
Recall the total information (compute) budget
\begin{equation}
R_{\mathrm{tot}}:=\sum_{\ell=0}^{L-1}R_\ell.
\label{eqC3_Rtot}
\end{equation}
Under the hypotheses of Theorem~\ref{thm7} and in the interior regime
(i.e., $\tilde W\succ 0$ and the optimizer satisfies $K^\star(R_{\mathrm{tot}})\prec \Sigma_X$),
the optimal value obeys the exponential law
\begin{equation}
D_{\mathrm{opt}}(R_{\mathrm{tot}})
=
k\,(\det\Sigma_X\,\det\tilde W)^{1/k}\,
\exp\!\left(-2\frac{R_{\mathrm{tot}}}{k}\right),
\label{eqC3_base}
\end{equation}
as stated in \eqref{eq88}. If $R_\ell=R_0$ for all $\ell$, then \eqref{eqC3_Rtot} yields
\begin{equation}
R_{\mathrm{tot}} = \sum_{\ell=0}^{L-1}R_0 = L R_0.
\label{eqC3_equal}
\end{equation}
Substituting \eqref{eqC3_equal} into \eqref{eqC3_base} yields
\[
D_{\mathrm{opt}}(L)
=
k(\det\Sigma_X\,\det\tilde W)^{1/k}\exp\!\left(-2L\frac{R_0}{k}\right),
\]
which is \eqref{eq89}. Let an inference procedure generate intermediate states $M_1,\dots,M_T$ such that for each step
\begin{equation}
I(X;M_t\mid M_{1:t-1})\le R_{\mathrm{step}}.
\label{eqC3_step}
\end{equation}
By the chain rule for mutual information and nonnegativity,
\begin{equation}
I(X;M_{1:T})
=
\sum_{t=1}^T I(X;M_t\mid M_{1:t-1})
\le
\sum_{t=1}^T R_{\mathrm{step}}
=
T R_{\mathrm{step}}
=: R_{\mathrm{tot}}.
\label{eqC3_chain}
\end{equation}
Therefore, the final representation $M_{1:T}$ satisfies the same information-budget constraint
as in Theorem~\ref{thm7} with total rate $R_{\mathrm{tot}}=TR_{\mathrm{step}}$.
Applying \eqref{eq88} with $R_{\mathrm{tot}}=TR_{\mathrm{step}}$ gives the same exponential scaling,
with $L$ replaced by $T$.
\end{proof}

\begin{remark}
The abstraction above yields several concrete interpretations for modern AI systems. 
\begin{itemize}
\item \emph{\bf Depth as cumulative rate:} If each layer can extract at most $R_\ell$ bits of new
task-relevant information from the context, then depth adds these budgets.
\item \emph{\bf Width and precision as per-layer rate:} A crude proxy is that $R_\ell$ scales with the
number of degrees of freedom that can be propagated through the layer: width, attention bandwidth,
and numerical precision. Lower precision reduces $R_\ell$, linking energy use and posterior entropy.
\item \emph{\bf Chain-of-thought as sequential rate allocation:} If each reasoning step adds only bounded
information about $X$, then Theorems~\ref{thm6} and~\ref{thm7} predict multiplicative reductions of
uncertainty with the number of steps.
\item \emph{\bf Retrieval and tool use as rate injection:} External retrieval and additional modalities
increase the effective $R_{\mathrm{tot}}$ by providing new informative inputs beyond the original $Z$.
\end{itemize}
\end{remark}

%%%%%%%%%%%%%%%%%%%%%%%%%%%%%%%%%%%%%%%%
%%%%%%%%%%%%%%%%%%%%%%%%%%%%%%%%%%%%%%%
%%%%%%%%%%%%%%%%%%%%%%%%%%%%%%%%%%%%%%%%%
\section{Discussion and Future Directions}
\label{section6}

This paper developed a unified information-theoretic framework for strategic Gaussian
RD under semantic misalignment, observation constraints, and information-rate
limitations. Beyond the explicit characterizations derived in
Sections~\ref{section3} and~\ref{section4}, the results raise broader conceptual implications for
energy-efficient and data-efficient artificial intelligence, and open several directions for
future research.

\subsection{Posterior Geometry as a Unifying Design Principle}

A central conceptual contribution of this work is the identification of posterior covariance
geometry as the fundamental object governing semantic performance under resource constraints.
Across all encoder observation models, the rate or compute constraint induces a log-det bound on
the posterior error covariance, while the semantic objective imposes a weighted precision
requirement. This leads to a geometric tradeoff between entropy reduction and semantic alignment,
expressed through generalized waterfilling laws in the direct and remote regimes and through
posterior design in the persuasion regime.

\subsection{Compute as an Information Rate: Implications for Efficient AI}

A key insight of the present work is that compute and architectural bottlenecks in modern learning
systems act as implicit information-rate constraints. Depth, width, attention bandwidth, context
length, and inference-time compute all limit the mutual information that can be propagated and
refined across layers or time steps. When viewed through this lens, neural scaling laws,
chain-of-thought reasoning, and iterative refinement can be interpreted as mechanisms for
allocating additional information rate to reduce posterior entropy.

\subsection{Multimodality as a Remedy for the Semantic Curse of Dimensionality}

Our gap analysis shows that remote semantic encoding suffers an irreducible geometric-mean penalty
that grows exponentially with dimension unless the semantic proxy is informative in every
direction. Multimodal observation emerges as a principled remedy. By aggregating complementary
modalities, the recoverable semantic covariance increases and the geometric penalty disappears,
allowing performance to approach the direct benchmark without proportional increases in rate or
compute.

\subsection{Strategic Communication and Alignment Considerations}

The strategic nature of the problem studied here highlights the importance of alignment and
objective mismatch in semantic systems. When the encoder and decoder optimize different
objectives, the encoder’s optimal strategy is not merely to compress information faithfully, but
to shape the decoder’s posterior beliefs in a manner consistent with its semantic priorities. In
the full-information regime, this manifests as Gaussian persuasion under a rate constraint.

\subsection{Modeling Assumptions and Limitations}

The analysis in this paper relies on several simplifying assumptions that enable closed-form
characterizations but also delineate the scope of the results.

First, we restrict attention to Gaussian sources and linear Gaussian semantic models with
quadratic objectives. This choice allows exact posterior characterizations (without the
concavification approach of Kamenica and Gentzkow \cite{kamenica2011bayesian}), but real-world
semantic variables may be discrete or heavy-tailed. Extending the posterior-design framework to
non-Gaussian settings is an important direction.

Second, the decoder is assumed to best respond via MMSE estimation. While this is optimal under
quadratic loss and Gaussian assumptions, practical systems may employ approximate inference.
Investigating robustness of the strategic RD solutions to suboptimal decoders would further
strengthen applicability.

Third, our compute-as-rate interpretation abstracts away many details of hardware and algorithmic
efficiency. Incorporating explicit energy models and hardware constraints remains open.

\section{Conclusion}
\label{section7}

This paper presented a unified theory of strategic semantic compression under rate, observation,
and compute constraints. By combining RD theory with information design, we
characterized fundamental limits on semantic inference in settings with misaligned objectives and
limited resources. The resulting framework yields explicit solutions for direct, remote, and
full-information encoding regimes, quantifies performance gaps between them, and explains how
multimodal observation and increased compute mitigate these gaps.

In closing, the central message of this work is that efficient intelligence is
fundamentally a problem of \emph{posterior design under resource constraints}.
By making posterior geometry explicit, information theory can play a central
role in guiding the development of future AI systems that are not only more
capable, but also more energy/data efficient and aligned.

\bibliographystyle{IEEEtran}

%\appendices

\bibliography{ref}

\end{document}